\documentclass[11pt]{article}

\usepackage{amsmath,amsthm,amssymb,amsfonts,verbatim}
\usepackage{fullpage}
\usepackage{todonotes}
\usepackage{footnote}
\usepackage[pagebackref]{hyperref}
\usepackage{arydshln}

 \providecommand{\F}{\mathbb{F}}

\date{}

\newtheorem{lemma}{Lemma}[section]
\newtheorem{theorem}[lemma]{Theorem}
\newtheorem{cor}[lemma]{Corollary}

\newtheorem{defn}{Definition}
\newtheorem{rmk}{Remark}

\theoremstyle{remark}

\renewcommand{\epsilon}{\varepsilon}
\renewcommand{\le}{\leqslant}
\renewcommand{\ge}{\geqslant}
\renewcommand{\leq}{\leqslant}
\renewcommand{\geq}{\geqslant}

\newcommand{\vnote}[1]{}

\def\PP{\mathbb{P}}

\def \Xi {{X^{[i]}}}

\newcommand{\Ga}{\alpha}
\newcommand{\Gb}{\beta}

\def \bc {{\bf c}}

\def \bx {{\bf x}}
\def \bh {{\bf h}}

\def \bg {{\bf g}}

\def \bu {{\bf u}}
\def \bv {{\bf v}}

\def\Supp {{\rm Supp }}

\title{\textbf{How long can optimal locally repairable codes be?}}
\author{Venkatesan Guruswami\thanks{Computer Science Department, Carnegie Mellon University, Pittsburgh, USA. This
work was done when the author was visiting the School of Physical and Mathematical Sciences, Nanyang Technological University, Singapore, and the Center of Mathematical Sciences and Applications, Harvard University. Research supported in part by NSF CCF-1563742. Email: {\tt venkatg@cs.cmu.edu}} \and Chaoping Xing\thanks{School of Physical and Mathematical Sciences, Nanyang Technological University, Singapore. Email: {\tt xingcp@ntu.edu.sg}.}
 \and Chen Yuan\thanks{Centrum Wiskunde \& Informatica, Amsterdam, Netherlands. Most of this work was done when the author was with the School of Physical and Mathematical Science, Nanyang Technological University, Singapore. Research supported in part by ERC H2020 grant No.74079 (ALGSTRONGCRYPTO). Email: {\tt Chen.Yuan@cwi.nl}}}

\begin{document}

\maketitle
\thispagestyle{empty}

\begin{abstract}
A locally repairable code (LRC) with locality $r$ allows for the recovery of any erased codeword symbol using only $r$ other codeword symbols. A Singleton-type bound dictates the best possible trade-off between the dimension and distance of LRCs --- an LRC attaining this trade-off is deemed \emph{optimal}. Such optimal LRCs have been constructed over alphabets growing linearly in the block length. Unlike the classical Singleton bound, however, it was not known if such a linear growth in the alphabet size is necessary, or for that matter even if the alphabet needs to grow at all with the block length.  Indeed, for small code distances $3,4$, arbitrarily long optimal LRCs were known over fixed alphabets.

\smallskip
Here, we prove that for distances $d \ge 5$, the code length $n$ of an optimal LRC over an alphabet of size $q$ must be at most roughly $O(d q^3)$. For the case $d=5$, our upper bound is $O(q^2)$. We complement these bounds by showing the existence of optimal LRCs of length $\Omega_{d,r}(q^{1+1/\lfloor(d-3)/2\rfloor})$ when $d \le r+2$. These bounds match when $d=5$, thus pinning down $n=\Theta(q^2)$ as the asymptotically largest length of an optimal LRC for this case.
\end{abstract}

\section{Introduction}
Modern distributed storage systems have been transitioning to erasure coding based schemes with good storage efficiency in order to cope with the explosion in the amount of data stored online. Locally Repairable Codes (LRCs) have emerged as the codes of choice for many such scenarios and have been implemented in a number of large scale systems e.g., Microsoft Azure~\cite{HSX+12} and Hadoop~\cite{SAP+13}.

A block code is called a locally repairable code (LRC) with locality $r$ if every symbol in the encoding is a function of
$r$ other symbols. This enables recovery of any single erased symbol in a local fashion by downloading at most $r$ other symbols. On the other hand, one would like the code to have a good minimum distance to enable recovery of many erasures in the worst-case. LRCs have been the subject of extensive study in recent years~\cite{HL07,GHSY12,PKLK12,SRKV13,JMX17,LMX17,FY14,PD14,TB14,TPD16,BTV17}.
LRCs offer a good balance between very efficient erasure recovery in the typical case in distributed storage systems where a single node fails (or becomes temporarily unavailable due to maintenace or other causes), and still allowing recovery of the data from a larger number of erasures and thus safeguarding the data in more worst-case scenarios.

A Singleton-type bound for locally repairable codes relating its
length $n$, dimension $k$, minimum distance $d$ and locality $r$ was first shown in the highly influential work~\cite{GHSY12}. It states that a linear locally repairable code $C$ must obey\footnote{The  bound in \cite{GHSY12} was shown even for a weaker requirement of locality only for the information symbols, but we focus on the more general all-symbol locality.}
\begin{equation}\label{eq:x1}
 d(C)\le n-k-\left\lceil \frac kr\right\rceil+2.
\end{equation}
%
Note that any linear code of dimension $k$ has locality at most $k$, so in the case when $r=k$ the above bound specializes to the classical Singleton bound $d \le n-k+1$, and in general it quantifies how much one must back off from this bound to accommodate locality.

A linear LRC that meets the bound \eqref{eq:x1} with equality is said to be an \emph{optimal} LRC. This work concerns the trade-off between alphabet size and code length for linear codes that are optimal LRCs.
Initially, the existence of such optimal LRCs and constructions were only known over fields that were exponentially large in the block length~\cite{HCL,SRKV13}.\footnote{If locality is desired only for the information symbols, then it is easy to construct optimal LRCs over linear-sized fields using any MDS code via the ``Pyramid" construction~\cite{HCL}. As we said, our focus is on LRCs with all-symbol locality which is more challenging to ensure.}
%
In a celebrated paper, Tamo and Barg \cite{TB14} constructed clever subcodes
of Reed-Solomon codes that yield a class of optimal locally repairable codes inheriting the field size $q \approx n$ of Reed-Solomon codes. This shows that one can have optimal LRCs with a field size similar to that of Maximum Distance Separable (MDS) codes which attain the classical Singleton bound $d=n-k+1$.

One is thus tempted to make an analogy between optimal LRCs and MDS codes. The famous MDS conjecture says that there are no non-trivial (meaning, distance $d > 2$) MDS codes with length exceeding $q+1$ where $q$ is its alphabet size, except in two corner cases  ($q$ even and $k=3$, or $k=q-1$) where the length is at most $q+2$. This conjecture was famously resolved in the case when $q$ is prime by Ball \cite{Ball}.

For optimal LRCs, it was shown that an analogous strong conjecture does not hold~\cite{LMX17} for almost every distance $d$  --- using elliptic curves, they gave LRCs length $q + 2\sqrt{q}$ (an earlier construction using rational function
fields achieved length $q+1$~\cite{JMX17}). A construction of length $n \approx \frac{r+1}{r} q$ was given for small distances in \cite{ZXL}. Note that all these constructions have length that is at most $O(q)$.

The MDS conjecture makes a very precise statement about the maximum possible length of MDS codes. An asymptotic upper bound of $n = O(q)$ (in fact even $n \le 2q$) is much easier to establish for MDS codes. Given this apparent parallel and the above-mentioned constructions which don't achieve code lengths exceeding $O(q)$, one might wonder if the Tamo-Barg result is asymptotically optimal, in the sense that optimal LRCs must have length at most $O(q)$. Rather surprisingly, it was not even known if $n$ must be bounded as a function of $q$ at all --- that is, it was conceivable that one could have arbitrarily long optimal LRCs over an alphabet of fixed size! Indeed, Barg et.al, \cite{BHHMV16} gave optimal LRCs using algebraic surfaces of length $n \approx q^2$ when the distance $d=3$ and $r \le 4$. This then inspired the discovery of optimal LRCs with \emph{unbounded length} for $d=3,4$ via cyclic codes~\cite{LXY}. We also include a simple construction of arbitrarily optimal LRCs for $d=3,4$ over any fixed field size that satisfies $q\ge r+1$.

\medskip\noindent \textbf{Our Results.}
Given this state of knowledge, the natural question that arises is whether there is any upper bound at all on the length of optimal locally repairable codes (as a function of its alphabet size).
In this paper, we answer this question affirmatively. In fact, we show that as soon as the distance $d \ge 5$, one cannot have unbounded length optimal LRCs (unlike the cases of $d =3,4$). Below is a statement of our upper bound on the code length of optimal LRCs. To the best of our knowledge, this is the first upper bound on the length of optimal LRCs.

\begin{theorem}[Upper bound on code length of LRCs]
\label{thm:intro-upperbound}
Let $d \ge 5$, and let $C$ be an optimal LRC with locality $r$ (that meets the bound \eqref{eq:x1} with equality) of length $n \ge \Omega(d r^2)$ over an alphabet of size $q$. Then $n \le O(d q^3)$ when $d$ is not divisible by $4$, and $n \le O(d q^{3+ 4/(d-4)})$ when $4 | d$.
\end{theorem}
Our actual upper bound is a bit better when $d \equiv 1 \pmod{4}$ and in particular yields $n \le O(q^2)$ when $d=5$. The technical condition that $n$ is at least $\Omega(d r^2)$ arises in ensuring that the code consists of $n/(r+1)$ \emph{disjoint} recovery groups of size $(r+1)$ each, that together ensure recoverability with locality $r$ for every codeword symbols.
Meanwhile, we have to point out that our bound yields nothing when $d$ is proportional to $n$. For this setting, we show
another bound that $d \le O(qr)$, showing that $d$ cannot be too large for an LRC with small locality $r$ unless the alphabet size is large. This follows from a combination of the puncturing argument in \cite{CM13} and the Plotkin bound.

In our second result, we complement the above result on the limitation of LRCs with a construction of super-linear (in $q$) length for $d \le r+2$.
\begin{theorem}[Construction of long LRCs]
\label{thm:intro-construction}
For every $r,d$ with $d \le r+2$, there exist optimal LRCs of length $n \ge \Omega_{d,r}(q^{1 + 1/\lfloor(d-3)/2\rfloor})$.\footnote{When $d=r+2$, it turns out that one cannot achieve bound \eqref{eq:x1} with equality; so we get codes with $d = n-k - \lceil k/r\rceil +1$ which is the optimal trade-off in this case. For $d \le r+1$ we attain \eqref{eq:x1} with equality.}
\end{theorem}

Again, to the best of our knowledge, this is the first code achieving super linear length in $q$ for $d \ge 5$. The previous best construction due to \cite{ZXL} achieved a length of  $\bigl( \frac{r+1}{r}\bigr) q$ for $d\leq r+1$. We establish Theorem~\ref{thm:intro-construction} via a greedy choice of the columns of the parity check matrix. Explicit constructions of such codes, as well as closing the gap between our upper and lower bounds on code length, are interesting questions for future work.

\medskip\noindent \textbf{Organization of the paper.}
The paper is organized as follows. In Section 2, we provide some preliminaries on locally repairable codes. In Section $3$, we prove our upper bounds on the length of optimal LRCs. In Section $4$, we present the greedy construction of an optimal LRC with super-linear length in its alphabet size.

\section{Preliminaries}
$[n]$ stands for $\{1,\ldots,n\}$. The floor function and ceiling function of $x$ are denoted by $\lfloor x \rfloor$ and $\lceil x \rceil$, respectively. An $[n,k,d]_q$ code is a linear code over the field of size $q$ that has length $n$, dimension $k$, and distance $d$.
We now define the local recoverability property of a code formally. We give this definition in general without assuming linearity, though we restrict our focus to linear codes in this paper.

\begin{defn}\label{def:1}{\rm Let $C$ be a $q$-ary block code of length $n$. For each $\Ga\in\F_q$ and $i\in \{1,2,\cdots, n\}$, define $C(i,\Ga):=\{\bc=(c_1,\dots,c_n)\in C\; : \; c_i=\Ga\}$. For a subset $I\subseteq \{1,2,\cdots, n\}\setminus \{i\}$, we denote by $C_{I}(i,\Ga)$ the projection of $C(i,\Ga)$ on $I$. For $i\in \{1,2,\cdots, n\}$,
a subset $R$ of $\{1,2,\dots,n\}$ that contains $i$ is a called a {\it recovery set} for $i$ if $C_{I_i}(i,\Ga)$ and $C_{I_i}(i,\Gb)$ are disjoint for any $\Ga\neq \Gb$, where $I_i=R\setminus\{i\}$.
Furthermore,  $C$ is called a locally recoverable code with locality $r$ if, for every $i\in \{1,2,\cdots, n\}$, there exists a recovery set $R_i$ for $i$ of size $r+1$.
}\end{defn}

\begin{rmk}{\rm The above definition of recovery sets is slightly different from that of recovery sets given in literature where $i$ is excluded in the recovery set $I_i$. The reason why we include $i$ in the recover set $R_i$ of $i$ is for convenience of proofs in this paper.
}\end{rmk}

For linear codes, which are the focus of this paper, the following lemma establishes a connection between
the locality and the dual code $C^{\perp}$. The proof is folklore, but we include it for the sake of completeness.
\begin{lemma}\label{dual_distance}
A subset $R$ of $\{1,2,\dots,n\}$ is a recovery set at $i$ of a linear code $C$ over $\F_q$ if and only if there exists a codeword in $C^{\perp}$ whose support contains $i$ and is a subset of $R$.
\end{lemma}
\begin{proof} Let  $G=(\bg_1,\ldots,\bg_n)\in \F_q^{k\times n}$ be a generator matrix  of $C$, where $\bg_j$ are column vectors of length $k$. Assume that $R$ is a  recovery set at $i$.
We prove the claim by contradiction. Suppose that there exists $i\in R$ such that  $C^{\perp}$ contains no codeword whose support contains $i$ and is a subset of $R$. This implies that $\bg_i$ is not a linear combination of $\{\bg_j\}_{j\in R\setminus\{i\}}$. Thus, $\bg_i\neq {\bf 0}$. If $\bg_j={\bf 0}$ for all $j\in R\setminus\{i\}$, then $C_{R\setminus\{i\}}(i,\Ga)\cap C_{R\setminus\{i\}}(i,\Gb)$ contains the zero vector for $\Ga,\Gb\in\F_q$. This is a contradiction to the definition of recovery sets.

Now assume that not all $\{\bg_j\}_{j\in R\setminus\{i\}}$ are the zero vector. Partition $R\setminus\{i\}$ into two disjoint sets $I$ and $J$ such that (i) the vectors $\{\bg_j\}_{j\in I}$ are linearly independent; and (ii) all vectors in $\{\bg_j\}_{j\in J}$ are linear combinations of $\{\bg_j\}_{j\in I}$. This implies that there exists a matrix $A$ of size $|I|\times |J|$ such that,  for every codeword $\bc\in C$, the projection $\bc_J$ of $\bc$ at $J$ is equal to $\bc_IA$. As  $\bg_i$ is not a linear combination of $\{\bg_j\}_{j\in R\setminus\{i\}}$, it follows that $\bg_i$ is not a linear combination of $\{\bg_j\}_{j\in I}$, i.e., $\bg_i$ and $\{\bg_j\}_{j\in I}$ are linearly independent.
 Thus, the set
$
\{\bx (\{\bg_j\}_{j\in I},\bg_i):\; \bx\in\F_q^k\}
$
is the entire space $\F_q^{|I|+1}$. This implies that, for any $\Ga,\Gb\in\F_q$ and $I\subset R\setminus\{i\}$, both the set $C_{I}(i,\Ga)$ and $C_{I}(i,\Gb)$ are equal to $\F_q^{|I|}$. Hence,
$C_{R\setminus\{i\}}(i,\Ga)=\left\{(\bu,\bu A):\bu \in\F_q^{|I|}\right\}=C_{R\setminus\{i\}}(i,\Gb)$.
This is a contradiction to the definition of recovery sets.

The other direction is obvious by the definition.
\end{proof}

For a $q$-ary $[n,k,d]$-linear LRC with locality
$r$, the Singleton-type bound says
\begin{equation}\label{eq:1}
d\leq n-k-\left\lceil \frac{k}{r}\right\rceil+2.
\end{equation}
Like the classical Singleton bound,  the Singleton-type bound \eqref{eq:1} does not take into account the cardinality of the code alphabet $q$. Augmenting this result, a recent work \cite{CM13} established a bound on the distance of locally repairable codes that depends on $q$, sometimes yielding better results.
However, in this paper, we specifically refer as optimal LRC a linear code achieving the bound \eqref{eq:1}. We now rewrite this bound in a form that will be more convenient to us.

\begin{lemma}\label{lem:2.1} Let $n,k,d,r$ be positive integers with $(r+1)|n$. If the Singleton-type bound \eqref{eq:1} is achieved, then
\begin{equation}\label{eq:2}
n-k=\frac{n}{r+1}+d-2-\left\lfloor\frac{d-2}{r+1}\right\rfloor.
\end{equation}
\end{lemma}
\begin{proof} Assume that  the Singleton-type bound \eqref{eq:1} is achieved, i.e., $d=n-k-\left\lceil \frac{k}{r}\right\rceil+2$.  Write $k=ar-b$ for some integers $a\ge 1$ and $0\le b\le r-1$. Then $d=n-k-a+2=n-(ar-b)-a+2$. This gives $a=\frac n{r+1}-\frac{d-2-b}{r+1}$. This also implies that $\frac{d-2-b}{r+1}$ is an integer. Therefore, we must have $\frac{d-2-b}{r+1}=\left\lfloor\frac{d-2}{r+1}\right\rfloor$ as $0\le b\le r-1$. The desired results follows as $$d=n-k-a+2=n-k-\frac{n}{r+1}+\frac{d-2+b}{r+1}+2=n-k-\frac{n}{r+1}+\left\lfloor\frac{d-2}{r+1}\right\rfloor+2\ . \qedhere $$
\end{proof}

\begin{rmk}
\label{rmk:d=r+2}
{\rm It turns out that the other direction of Lemma \ref{lem:2.1} is also true if $d-2\not\equiv r\pmod{r+1}$.
 }
\end{rmk}



\section{An Upper Bound on Code Lengths}
In this section, we investigate the  upper bound on the code lengths of optimal LRCs over a finite field $\F_q$.
For simplicity, we assume that $n$ is divisible by $r+1$ throughout this section.
However, in Remark \ref{rmk:divisible} and \ref{rmk:nondiv}, we extend our results to cover the cases when $n$ is not divisible by $r+1$.
We end this section by another upper bound that handles the case $d=O(n)$.

\subsection{Justifying the assumption of disjoint recovery sets}

We first argue that a $r$-local LRC with block length $n$ divisible by $r+1$ can be assumed, under modest conditions on the parameters, to contain $n/(r+1)$ disjoint recovery sets that each allow for recovery of $(r+1)$ codeword symbols. This structure will then be helpful to us in upper bounding the length of LRCs.

We remark that the structure theorem in~\cite{GHSY12} showed that the information symbols can be arranged into $k/r$ disjoint groups each with a local parity check, under the assumption that $r | k$. However, we seek all-symbol's locality, and their argument does not directly apply.


%
\begin{lemma}\label{lem:disjoint}
Let $C$ be an $[n,k,d]_q$ linear optimal LRC with locality $r$. Then, there exist $\frac{n}{r+1}$ disjoint recovery sets, each of size $r+1$ provided that
\begin{equation}\label{eq:3}\frac{n}{r+1}\geq \left(d-2-\left\lfloor\frac{d-2}{r+1}\right\rfloor\right)(3r+2)+\left\lfloor\frac{d-2}{r+1}\right\rfloor+1.\end{equation}
\end{lemma}
\begin{proof} Put $h=d-2-\lfloor\frac{d-2}{r+1}\rfloor$. Then
$n-k=\frac{n}{r+1}+h$ and $h\leq d-2$. Lemma~\ref{lem:2.1} implies that the parity-check matrix of
code $C$ must has size $(\frac{n}{r+1}+h)\times n$.
We now construct a parity-check matrix of $C$ as follows. First, we arbitrarily choose $i_1\in\{1,2,\dots,n\}$. By Lemma \ref{dual_distance}, there is a codeword $\bc_1$ of $C^\perp$ such that $\Supp(\bc_1)$ contains $i_1$ and has size at most $r+1$. Put $R_1=\Supp(\bc_1)$ and choose $i_2\in\{1,2,\dots,n\}\setminus R_1$.  By Lemma \ref{dual_distance} again, there is a codeword $\bc_2$ of $C^\perp$ such that $\Supp(\bc_2)$ contains $i_2$ and has size at most $r+1$. Put $R_2=\Supp(\bc_1)$ and choose $i_3\in\{1,2,\dots,n\}\setminus (R_1\cup R_2)$. Continue in this fashion to get $\ell$ codewords $\bc_i\in C^\perp$ and $R_i=\Supp(\bc_i)\subset \{1,2,\dots,n\}$  for $1\le i\le \ell$ such that $R_1,R_2,\dots,R_\ell$ are pairwise distinct and $\cup_{i=1}^\ell R_i=\{1,2,\dots,n\}$. As $\bc_1,\bc_2,\dots,\bc_\ell$ are linearly independent, we have $\frac{n}{r+1}\le \ell\le n-k=\frac{n}{r+1}+h$. It is clear that $\ell=\frac{n}{r+1}$ if and only if the sets $R_1,R_2,\dots,R_\ell$ are pairwise disjoint.

We claim that  $\ell$ must be equal to $\frac{n}{r+1}$. Then our desired result follows. Suppose that this claim is not true, i.e, $\ell\ge \frac{n}{r+1}+1$.
Assume that there are $a$ recovery sets with size less than $r+1$. Since the union of $R_1,\ldots,R_\ell$  is $\{1,2,\dots,n\}$ , we have $ar+(\ell-a)(r+1)\geq n$, i.e., $a\leq \ell (r+1)-n$. Hence, the number of $(r+1)$-sized recovery sets satisfies $\ell-a\geq n-\ell r\geq \frac{n}{r+1}-hr$.
Without loss of generality, we may assume that $R_1,\ldots,R_{\ell_1}$ with $\ell_1= \frac{n}{r+1}-hr$ are of size $r+1$. We are going to show that there are at least $b:=\lfloor\frac{d-2}{r+1}\rfloor+1$ pairwise disjoint sets, say $R_1,R_2,\ldots, R_b$, among $R_1,\ldots,R_{\ell_1}$ such that $R_i\cap R_j=\emptyset$ for all $1\le i\le b$ and $1\le j \le \ell$, $j\neq i$.
Since $\cup_{i=1}^\ell R_i$ is $\{1,2,\dots,n\}$, we have
$$
\sum_{j=1}^\ell|R_j|-|\bigcup_{j=1}^\ell R_j|=\sum_{\alpha=1}^{n}(\sum_{j=1,\alpha \in R_j}^{\ell} 1-1)
$$
If $\sum_{j=1,\alpha \in R_j}^{n} 1>1$ for some $\alpha$, we remove all the sets $R_j$ that contain this $\alpha$.
Note that
$$
\sum_{j=1}^\ell|R_j|-|\bigcup_{j=1}^\ell R_j|\leq\left(\frac n{r+1}+h\right)(r+1)-n=h(r+1)
$$
This implies that we remove at most $2h(r+1)$ sets from $R_1,\ldots,R_{\ell_1}$. Let $R_1,R_2,\ldots, R_e$ be the sets left after this operation.
From our argument, we know that if $\alpha \in R_j$ for $1\le j \le e$, then $\alpha$ does not belong to any other set in $R_1,\ldots,R_{\ell}$.
This implies our requirement that $R_i\cap R_j=\emptyset$ for all $1\le i\le e$ and $1\le j \le \ell$,$j\neq i$.
It remains to lower bound $e$. Since we remove at most $2h(r+1)$ sets from $R_1,\ldots,R_{\ell_1}$, we have
\begin{equation}\label{eq:5}
e \ge \ell_1-2h(r+1)=\frac{n}{r+1}-h(3r+2)\ge \left\lfloor\frac{d-2}{r+1}\right\rfloor+1 = b \ , \end{equation}
where the inequality in \eqref{eq:5}  is due to the
condition given in \eqref{eq:3}.

%



Now let us construct a parity-check matrix $H$ as follows. We take the first $\ell$ rows to be $\bc_1,\bc_2,\dots,\bc_\ell$ (note that these vectors are linearly independent). Then extend arbitrarily to an parity-check matrix $H$ of size $(n-k)\times n$. We reconsider the submatrix $H_1$ of $H$ consisting of the first $b(r+1)$ columns. Then it must have the following form
\begin{equation}\label{eq:matrix}
H_1=\left(
  \begin{array}{c}
    \bv_1 \\
    \bv_2 \\
    \vdots \\
    \bv_b \\
  O  \\
    H_2 \\
  \end{array}
\right),
\end{equation}
where $\bv_1,\ldots,\bv_b\in \F_q^{b(r+1)}$, $O$ is a $(\ell-b)\times b(r+1)$ zero matrix and $H_2$ is a $(\frac{n}{r+1}+h-\ell)\times b(r+1)$ matrix. Note that $b(r+1)\ge d-1$ and $C$ has minimum distance $d$. This implies
 any $d-1$ columns of $H_1$ are linearly independent or equivialently
the rank of $H_1$ is at least $d-1$. However, the number of nonzero rows of $H_1$ is at most
$$b+\frac{n}{r+1}+h-\ell\leq b+h-1=\left\lfloor\frac{d-2}{r+1}\right\rfloor+1+d-2-\left\lfloor\frac{d-2}{r+1}\right\rfloor-1=d-2.$$
The first inequality follows from the assumption $\ell\geq \frac{n}{r+1}+1$. This contradiction concludes that $\ell=\frac{n}{r+1}$ and the desired result follows.
\end{proof}


\begin{rmk}\label{rmk:divisible}
{\rm
A similar result is still hold when $n$ is not divisible by $r+1$.
In this case, the minimum number of recovery sets covering $n$ indices becomes $\lceil \frac{n}{r+1} \rceil$.
%
Let us start from the code meeting the Singleton-type bound,
$$
d=n-k-\lceil\frac{k}{r}\rceil+2\geq n-k-\frac{k}{r}-1+2.
$$
It follows that
$$
k\geq \frac{r}{r+1}(n-(d-2)-1),
$$
and then
$$
n-k\leq \frac{n}{r+1}+d-2-\frac{d-2}{r+1}+\frac{r}{r+1}.
$$
Since $n-k$ is an integer, we have
$$n-k\leq \lceil \frac{n}{r+1} \rceil+d-2-\lfloor \frac{d-2}{r+1} \rfloor.$$
The fact that the number of
recovery sets covering all indices is at least $\lceil \frac{n}{r+1} \rceil$ leads to $h=n-k-\lceil\frac{n}{r+1}\rceil\leq d-2-\lfloor \frac{d-2}{r+1} \rfloor$.
The rest of the proof is the same.  }
\end{rmk}
\subsection{Proving the upper bound}
In this subsection, we prove Theorem~\ref{thm:intro-upperbound} (restated more formally below) that gives an upper bound on the length $n$ of a LRC in terms of its alphabet size $q$. The parity check view of an LRC will be instrumental in our argument. We will make use of Lemma~\ref{lem:disjoint} and the classical Hamming upper bound on the size of codes as a function of minimum distance to derive our result.

\begin{theorem}\label{thm:upperbound}
Let $C$ be an optimal $[n,k,d]_q$-linear locally repairable codes of locality $r$ with $(r+1) | n$ and parameters satisfying the inequality \eqref{eq:3} given in Lemma~ {\rm \ref{lem:disjoint}}.
If $d\ge 5$ and $d\equiv a\pmod{4}$ for some $1\le a\le 4$, then
\begin{equation}n=\left\{ \begin{array}{ll}O(dq^{\frac{4(d-2)}{d-a}-1}) &\mbox{if $a=1,2$},\\
O(dq^{\frac{4(d-3)}{d-a}-1})&\mbox{if $a=3,4$.}
\end{array}\right.\end{equation}  In particular, we have
$n=O\left(dq^{3+\frac{4}{d-4}}\right)$. Furthermore, we have $n=O(q^2)$, $O(q^3)$, $O(q^3)$, $O(q^4)$, $O(q^{2.5})$ and $O(q^3)$ for $d=5,6,7,8,9,$ and $10$, respectively.
\end{theorem}
\begin{proof} Again we
let $n-k=\frac{n}{r+1}+h$ with $h=d-2-\lfloor\frac{d-2}{r+1}\rfloor\leq d-2$.
By Theorem~\ref{lem:disjoint}, we know that there exist $\ell:=\frac{n}{r+1}$ codewords $\bc_1,\ldots,\bc_\ell$ of $C^\perp$ such that the supports $\Supp(\bc_1),\ldots,\Supp(\bc_\ell)$, each of size $r+1$, are pairwise disjoint. Put $R_i=\Supp(\bc_i)$. By considering an equivalent code, we may assume that $R_i=\{(i-1)(r+1)+1,\ldots,i(r+1)\}$ for $i=1,2,\dots,\ell$ and the projection of $\bc_i$ at $R_i$ are equal to all-one vector ${\bf 1}$ of length $r+1$.

The parity-check matrix $H$ has the following form
\begin{equation}\label{eq:paritycheck}
H=\left(
  \begin{array}{c}
    \begin{array}{c|c|ccc|c}
      {\bf 1} & \mathbf{0} & \cdots & \cdots & \cdots &\mathbf{0} \\
      \mathbf{0} & {\bf 1} & \cdots &\cdots &  \cdots & \mathbf{0} \\
      \vdots & \vdots & \ddots &\ddots & \ddots &\vdots \\
      \mathbf{0} & \mathbf{0} & \cdots & \cdots  & \cdots & {\bf 1}
    \end{array}
     \\ \hdashline
    \huge{A} \\
  \end{array}
\right),
\end{equation}
where $A$ is an $h\times n$ matrix over $\F_q$.
The submatrix consisting of the first $\ell$ rows of $H$ is a block diagonal matrix.
Let $\bh_{i,j}$ be the $(i(r+1)+j)$-th column of $H$, i.e.,
\begin{equation}\label{eq:columnvector}
\bh_{i,j}=(\underbrace{0,\ldots,0}_{i-1},1,\underbrace{0,\ldots,0}_{\ell-i},\bv_{i,j})^T
\end{equation}
for some $\bv_{i,j}\in \F_q^{h}$, where $T$ stands for transpose.

Define
$$\bh'_{i,j}:=\bh_{i,j}-\bh_{i,r+1}=
(\underbrace{0,\ldots,0}_\ell,\bv_{i,j}-\bv_{i,r+1})^T$$
for $i\in [\ell]$ and $j\in [r]$.
We claim that any $\lfloor\frac{d-1}{2}\rfloor$ of $\bh'_{1,1},\ldots,\bh'_{\ell,r}$ are linearly independent.
Indeed, for any $t:=\lfloor\frac{d-1}{2}\rfloor$ vectors $\bh'_{i_1,j_1},\ldots,\bh'_{i_t,j_t}$ and  scalars $\lambda_{i_1,j_1},\ldots,\lambda_{i_t,j_t}\in \F_q$ satisfying
$
\sum_{k=1}^{t}\lambda_{i_k,j_k}\bh'_{i_k,j_k}={\bf 0}$, i.e.,
$\sum_{k=1}^{t}\lambda_{i_k,j_k}(\bh_{i_k,j_k}-\bh_{i_k,r+1})={\bf 0}$, we have $\sum_{k=1}^{t}\lambda_{i_k,j_k}\bh_{i_k,j_k}-\sum_{k=1}^{t}\lambda_{i_k,j_k}\bh_{i_k,r+1}={\bf 0}$.

Note that $\bh_{i_1,j_1},\ldots,\bh_{i_k,j_k}$ together with $\bh_{i_1,r+1},\ldots,\bh_{i_k,r+1}$ are at most $2t\leq d-1$ distinct
columns of $H$. It follows that they are linearly independent and thus the coefficient $\lambda_{i_1,j_1},\ldots, \lambda_{i_t,j_t}$ must be all zero.

 Moreover, we note that the first $\ell$ components of $\bh'_{i,j}$ are all zero for
$(i,j)\in [\ell]\times [r]$.
We shorten the vector $\bh'_{i,j}$ by puncturing its first $\ell$ coordinates. Denote by $\widetilde{\bh}_{i,j}$ the shortened vectors.
It is clear that any $\lfloor\frac{d-1}{2}\rfloor$ of $\widetilde{\bh}_{1,1},\ldots,\widetilde{\bh}_{\ell,r}$ are still linearly independent.
Let $H_2$ be the matrix whose columns consists of
$\widetilde{\bh}_{i,j}$ for $i=1,\ldots,\ell$ and $j=1,\ldots,r$ and let $C_2$ be a linear code whose parity-check matrix
is $H_2$. Then $C_2$ is a linear code with length $N:=n-\ell=\frac{rn}{r+1}$, dimension at least $N-h$ and distance at least $\lfloor\frac{d-1}{2}\rfloor+1$. We now apply the Hamming bound to  $C_2=\left[n-\ell, \ge n-\ell-h,\ge \lfloor\frac{d-1}{2}\rfloor+1\right]$-linear code.

Let $d=4d_1+a$ for some $d_1\ge 1$ and $1\le a\le 4$.

{\it Case 1.} $a=1$ or $2$. In this case, we have $\lfloor\frac{d-1}{2}\rfloor+1=2d_1+1$.
Applying  the Hamming bound to $C_2$ gives
$$
q^{N-h}\leq \frac{q^N}{\sum_{i=1}^{d_1}\binom{N}{i}(q-1)^{i}}\leq  \frac{q^N}{\binom{N}{d_1}(q-1)^{d_1}}\leq \frac{q^N}{(\frac{N}{d_1})^{d_1}(q-1)^{d_1}},
$$
i.e.,  $\frac{rn}{r+1}=N\leq \frac{d_1}{q-1}\times q^{\frac{h}{d_1}}=\frac{d-a}{4(q-1)}\times q^{\frac{4h}{d-a}}\le \frac{d-a}{4(q-1)}\times q^{\frac{4(d-2)}{d-a}}$. The last inequality follows from the fact that $h\leq d-2$.

{\it Case 2.} $a=3$ or $4$. In this case, we have $\lfloor\frac{d-1}{2}\rfloor+1=2d_1+2$. Deleting  the first coordinate of $C_2$ gives a $q$-ary $[N-1,N-h,\ge 2d_1+1]$-linear code.
Applying  the Hamming bound to $[N-1,N-h,\ge 2d_1+1]$ gives
$$
q^{N-h}\leq \frac{q^{N-1}}{\sum_{i=1}^{d_1}\binom{N-1}{i}(q-1)^{i}}\leq  \frac{q^{N-1}}{\binom{N-1}{d_1}(q-1)^{d_1}}\leq \frac{q^{N-1}}{(\frac{N-1}{d_1})^{d_1}(q-1)^{d_1}},
$$
i.e.,  $\frac{rn}{r+1}-1=N-1\leq \frac{d_1}{q-1}\times q^{\frac{h-1}{d_1}}=\frac{d-a}{4(q-1)}\times q^{\frac{4(h-1)}{d-a}}\le \frac{d-a}{4(q-1)}\times q^{\frac{4(d-3)}{d-a}}$.
In conclusion, we have
\[n\le\left\{ \begin{array}{ll}\frac{r+1}r\times \frac{d-a}{4(q-1)}\times q^{\frac{4(d-2)}{d-a}} &\mbox{if $a=1,2$},\\
\frac{r+1}r\left(\frac{d-a}{4(q-1)}\times q^{\frac{4(d-3)}{d-a}}+1\right)&\mbox{if $a=3,4$.}
\end{array}\right.\]
The desired result follows.
\end{proof}
\begin{rmk}\label{rmk:nondiv}
{\rm
Let us extend this result to the case $n$ is not divisible by $r+1$.
From Remark \ref{rmk:divisible}, we obtain $\lceil \frac{n}{r+1} \rceil$ recovery sets $R_1,\ldots,R_{\lceil \frac{n}{r+1} \rceil}$
covering all of the $n$ indices.
There are at most $(r+1)\lceil \frac{n}{r+1} \rceil-n \le r$ indices that belong to more than $1$ of these $\lceil \frac{n}{r+1} \rceil$ recovery sets.
We first build the parity-check matrix $H$ whose first $\lceil \frac{n}{r+1} \rceil$ rows are $\bc_1,\ldots,\bc_{\lceil \frac{n}{r+1} \rceil}$ where $\bc_i$ corresponds to recovery set $R_i$. Then, we remove the columns from $H$ whose indices belong to multiple recovery sets. After removing at most $r$ columns, we apply the same argument to the resulting matrix. It is thus clear that the same result also holds for the case $n$ is not divisible by $r+1$, with a small adjustment of $r$ in the final upper bound on the code length.
}\end{rmk}

\begin{rmk}{\rm
From our proof of Theorem \ref{thm:upperbound}, one might see why our argument is not applicable to the optimal LRC with distance less than $5$. In our argument, the optimal LRC of distance $d$ is reduced to a code of distance at least $\lfloor\frac{d+1}{2}\rfloor$ without locality. If $d\leq 4$, this reduced code might be the Hamming code whose code length is independent of the alphabet size. That explains the reason why our argument fails in this scenario. On the other hand, there indeed exists unbounded length of optimal LRCs of distance $d\leq 4$~\cite{LXY}. Therefore, our argument reveals the inherent differences of optimal LRCs with distance less than $5$ and above.}
\end{rmk}

Note that Theorem \ref{thm:upperbound} says nothing when $d$ is proportional to $n$. To obtain a meaningful upper bound in this case, we resort to Theorem $1$ in \cite{CM13}.

\begin{theorem}
The minimal distance of optimal LRCs is upper bounded by $q\frac{r^2+2r+3}{r}$.
\end{theorem}
\begin{proof}
By Theorem $1$ in \cite{CM13}, the dimension $k$, locality $r$ and minimal distance $d$ of optimal LRCs must obey
\begin{equation}\label{eq:dimension}
k \le tr + k_{max}(n-t(r+1), d)
\end{equation}
where $k_{\max}(m,e)$ is the largest dimension of a linear code in $\F_q^m$ of distance $e$, and $t$ is an arbitrary integer parameter, $0 \le t \le n/(r+1)$.

Pick $t$ to be $\lceil\frac{n-(1-\epsilon)\frac{qd}{q-1}}{r+1}\rceil$ so that $n-t(r+1) \le (1-\epsilon) \frac{qd}{q-1}$.
The Plotkin bound now gives $k_{\max}(n-t(r+1),d) \le \log_q (1/\epsilon)$.
Set $\epsilon=1/q^2$ and \eqref{eq:dimension} gives that $k\leq tr+2$.
On the other hand, the Singleton-type bound says that
$$
k\ge n-\frac{n}{r+1}-d+2+\lfloor\frac{d-2}{r+1}\rfloor\ge n-\frac{n}{r+1}-d+2+\frac{d-2}{r+1}-1.
$$
This implies that
$$
r\left(\frac{n-(1-\frac{1}{q^2})\frac{qd}{q-1}}{r+1}+1\right)\ge tr\ge n-\frac{n}{r+1}-d+\frac{d-2}{r+1}-1.
$$
Solving this inequality in $d$ gives us
$$
d\leq q\left( \frac{r^2+2r+3}{r} \right) \ .
\qedhere
$$
\end{proof}
The following Corollary is an immediate consequence.
\begin{cor}\label{cor:proportional}
Assume that $d=O(n)$ and $r$ is a constant, then the length $n$ of optimal LRCs is upper bounded by $O(q)$.
\end{cor}

\section{Construction of LRCs of super-linear length}

To the best of our knowledge, all known constructions of optimal LRCs have block length $n\le O(q)$ unless $d\leq 4$.
Our upper bound in the preceding section implies that $n$ must be upper bounded by (roughly) $q^3$.
A natural question arises whether
there exists optimal LRC with super linear length in $q$, e.g, $n=\Omega(q^{1+\epsilon})$ and some constant $d>4$.
In this section we answer this question affirmatively, showing such codes for all $d \le r+2$.

When $d=r+2$ and $r+1|n$, the  Singleton-type bound \eqref{eq:x1} can't be met~\cite[Corollary 10]{GHSY12}. In this case, by an optimal LRC we mean a code attaining the trade-off
$d=n-k-\lceil \frac{k}{r} \rceil +1$.
When $n$ is not divisible by $r+1$, by shortening the code, it is still possible to obtain the optimal LRCs.
We leave this discussion to Corollary \ref{cor:shorten} and \ref{cor:opt}.

Before stating our main results, we notice a simple but useful fact, i.e., the locality and minimum distance of a linear code can be reflected by representing its generator matrix properly.
Thus, it is sufficient to concentrate on the construction of generator matrix.
As a warmup, let us begin with the generator matrix of optimal LRCs with minimum distance $3$ and $4$.

\begin{theorem}
Assume that $d=3,4$, $d-2\leq r$ and $r+1|n$, there exist optimal LRCs of arbitrarily lengths as long as $q\geq r+1$.
\end{theorem}
\begin{proof}
For $d=3,4$, $d-2\leq r$ and $r+1|n$, the Singleton-type bound implies that $n-k=\frac{n}{r+1}+d-2$.
Since $q\geq r+1$, we let $A$ be a $(d-2)\times (r+1)$ Vandermonde matrix over $\F_q$ such that
$$
A_1=\left(
  \begin{array}{c}
    \mathbf{1} \\
    A \\
  \end{array}
\right)
$$
is a $(d-1)\times (r+1)$ Vandermonde matrix.
Define $(\frac{n}{r+1}+d-2)\times n$ matrix
\begin{equation*}
H=\left(
    \begin{array}{c|c|c|c}
      {\bf 1} & \mathbf{0}   & \cdots &\mathbf{0} \\
      \mathbf{0} & {\bf 1} &  \cdots & \mathbf{0} \\
      \vdots & \vdots  & \ddots &\vdots \\
      \mathbf{0} & \mathbf{0}   & \cdots & {\bf 1}\\
      A & A &  \cdots & A
    \end{array}
\right)
\end{equation*}
Where ${\bf 1}$ and $\mathbf{0}$ are all-$1$ and all-$0$ vectors in $\F_q^{r+1}$, respectively.
We partition the columns of $H$ into $\frac{n}{r+1}$ blocks $B_1,\ldots,B_{\frac{n}{r+1}}$ such that
$\bh \in B_i$ if its $i$-th component is non-zero. From the expression of matrix $H$, it is clear that
each column belongs to exactly one block and the columns in distinct blocks are linearly independent.
Moreover, any $d-1$ columns in the same block are linearly independent due to the property of Vandermonde matrix $A_1$.
Next we show that any $d-1$ columns of $H$ are linearly independent.
It suffices to verify this claim for the case $d=4$.
To see this, we pick any three columns $\bh_i,\bh_j,\bh_t$ from $H$.
Nothing needs to prove if these three columns belong to the same block.
We assume that they belong to at least two blocks.
Without loss of generality, $\bh_t$ is in a block that does not contain $\bh_i$ and $\bh_j$.
From above observation, we see that $\bh_t$ is linearly independent from $\bh_i$ and $\bh_j$. It is clear $\bh_i$ and $\bh_j$ are linearly
independent no matter whether they belong to the same block or different blocks.
Thus, any $3$ columns of $H$ are linearly independent.
Let $C$ be the linear code whose parity-check matrix is $H$. It is clear that $C$ has length $n$, dimension $k(C)\geq n-\frac{n}{r+1}-(d-2)=\frac{rn}{(r+1)}-(d-2)$, distance $d(C)\geq d$ and locality $r$.
The condition $d-2\leq r$ leads to $\lceil\frac{k(C)}{r}\rceil\geq \frac{n}{r+1}$ and thus $k(C)+\lceil\frac{k(C)}{r}\rceil\geq n-(d-2)$.
The desired result follows since $d(C)\geq d\geq n-k(C)-\lceil\frac{k(C)}{r}\rceil+2$.
\end{proof}

Next, we proceed to our main result of this section, the construction of optimal LRCs of super-linear length for $d\leq r+2$. Like the case
$d=3,4$, it is sufficient to construct a generator matrix of these codes.

\begin{theorem}\label{thm:construction}
Assume $d\leq r+2$ and $(r+1)|n$. There exist optimal LRCs of length $n=\Omega_{d,r}(q^{1+\frac{1}{\lfloor (d-3)/2\rfloor}})$. In particular, one obtains the best possible length $n=O(q^2)$ for optimal LRC of minimum distance $5$ if $r\ge 3$ and $(r+1)|n$.
\end{theorem}

\begin{proof}
Let $n=\eta q^{1+1/\lfloor (d-3)/2 \rfloor}$ with some constant $\eta$ that only depends on $d$ and $r$, i.e., $\eta=\Omega_{d,r}(1)$. We will determine $\eta$ later.
It suffices to construct a matrix $H$ and show that the code $C$ derived from this parity-check matrix is an optimal
LRC. Label and order the $n$ coordinates with
$(i,j)\in [\frac{n}{r+1}]\times[r+1]$, i.e., $(i_1,j_1)$ precedes $(i_2,j_2)$ if $i_1<i_2$ or $i_1=i_2$ and $j_1<j_2$.
Let $H=(\bh_{i,j})_{(i,j)\in [\frac{n}{r+1}]\times[r+1]}$ where $\bh_{i,j}\in \F_q^{n-k}$. That means $H$ consists of the columns
$\bh_{i,j}$ for $(i,j)\in [\frac{n}{r+1}]\times[r+1]$.
We start from $\bh_{1,1}$ and determine the value of $\bh_{i,j}$ column by column in the above order.
In each step, we make sure that the new column $\bh_{i,j}$ together with any $d-2$ columns preceding the $(i,j)$-th column are linearly independent.
Meanwhile, the matrix $H$ holds the same form\footnote{The same form is referred to that their distributions of non-zero entry in upper half matrix (matrix lying above $A$) are the same, i.e., entry of value $1$ and $0$ represents the nonzero entry and zero, respectively.} as the matrix in
\eqref{eq:paritycheck}. If we can achieve both of the conditions, we are done. Define $\frac{n}{r+1}$ blocks $B_1,\ldots,B_{\frac{n}{r+1}}$ such that $B_i=\{\bh_{i,1},\ldots,\bh_{i,r+1}\}$.
That means we partition the $n$ columns into $\frac{n}{r+1}$ disjoint blocks. Algorithm $1$ below gives the iterative method to compute the columns $\bh_{i,j}$'s.

\medskip
\begin{savenotes}
\begin{center}
\fbox{\begin{minipage}{35em}
\begin{center} {\bf Algorithm $1$} \end{center}
\begin{itemize}
\item For $i=1, \ldots, \frac{n}{r+1}$, and $j=1,\ldots,r+1$, do the following operation.
\begin{itemize}
\item Find $\bv \in \F_q^{n-k}$ of form \eqref{eq:columnvector}\footnote{Only the $i$-th component out of the first $\frac{n}{r+1}$ components is nonzero.} such that $\bv$ is linearly independent of any subset of at most $(d-2)$ columns $\bh_{i,j}$ chosen before this step.
\item Let $\bv$ be the $(i,j)$-th column of $H$, i.e., $\bh_{i,j}=\bv$.
\end{itemize}
\end{itemize}

\end{minipage}}
\end{center}
\end{savenotes}
\bigskip

We justify Algorithm $1$ by showing that there always exists such $\bh_{i,j}$ for any $(i,j)\in [\frac{n}{r+1}]\times[r+1]$.
Assume that we arrive at the $(a,b)$-th column.
If $b=1$, the construction is trivial. Let $\bh_{a,b}$ be a column vector such that
the first $\frac{n}{r+1}$ components except $i$-th component are zero. Obviously, it matches the form of Equation~\ref{eq:columnvector}. The linearly independence is also trivial since the $i$-th component of all the columns $\bh_{i,j}$ for $i<a$ is $0$.
Otherwise, to simplify our discussion, we assume that the first $d-2$ columns are already found.
Since any $d-2$ columns prior to the $(a,b)$-th column are already linearly independent by our algorithm, it suffices to show that
$\bh_{a,b}$ is linearly independent from these $d-2$ columns.
To achieve this, we need to check all possible combinations of
these $d-2$ columns.
Assume that these $d-2$ columns are chosen exactly from $t$ blocks. Obviously, block $B_a$ must be selected. Otherwise, the same reason for $b=1$ implies that $\bh_{a,b}$ is linearly independent of these $d-2$ columns.
Without loss of generality, we assume that these $t$ blocks are $B_1,\ldots,B_{t-1}$ and $B_a$ and there are $i_j$ columns picked from block $B_j$. Then, the submatrix $H_1$ consisting of these $d-2$ columns has the following form:
$$
H_1=\left(
  \begin{array}{c}
    \begin{array}{c|c|c|c|c}
      \bx_1 & \mathbf{0}  &\cdots & \mathbf{0}&\mathbf{0}\\
      \mathbf{0} & \bx_2 &    \cdots & \mathbf{0} &\mathbf{0}\\
      \vdots & \vdots &    \ddots &\vdots &\vdots\\
      \mathbf{0} & \mathbf{0}  & \cdots &  \bx_{t-1} &\mathbf{0}  \\
     \vdots & \vdots &    \ddots &\vdots &\vdots\\
       \mathbf{0} & \mathbf{0} &  \cdots & \mathbf{0} & \bx_{a}   \\
       \vdots & \vdots  & \ddots &\vdots  &\vdots \\
    \end{array}\\
      \hdashline
    \huge{A_1}
  \end{array}
\right)
$$
where $\bx_i\in \F_q^{i_j}$ and $A_1$ is a $(d-2)\times (d-2)$ matrix.
If any $B_j$, $j=1,2,\dots,t-1$, contains only one column, then that column is linearly independent of the rest of the $d-3$ columns and $\bh_{a,b}$, and therefore can be removed from consideration. Thus we may assume that there are at least two columns chosen in each block except block $B_a$.
Thus, $t$ is at most ${\lfloor\frac{d-1}{2}\rfloor}$. Recall that our goal is to ensure that $\bh_{a,b}$ is linearly independent of these at most $d-2$ columns in total chosen from the blocks $B_1,\dots,B_{t-1}$.
Given the $t$ blocks and the $d-2$ columns chosen from them, we count the number of bad $\bh_{a,b}$ which are linear combinations of these $d-2$ columns. If the number of such linear combinations is smaller than the size of the whole space of possible choices of $\bh_{a,b}$, we are done.
To achieve this, we need to determine the maximal subspace $V$ spanned by these $d-2$ columns such that all the vectors in $V$ has the same form as the
vector $\bh_{a,b}$, i.e., the first $\frac{n}{r+1}$ components except $a$-th component are zero.

For block $B_j$ with $j\neq a$, by the expression of matrix $H_1$, the $i_j$ columns of $B_j$ created a $i_j-1$-dimensional subspace where the first $\frac{n}{r+1}$ components of all the vectors are $0$.
That means, block $B_j$ for $j\neq a$ contributes $i_j-1$ linearly independent vectors to the maximal subspace $V$.
For block $B_a$, it contributes at most $i_a$ linearly independent vectors to the maximal subspace $V$.
It follows that the dimension of $V$ is at most
$\sum_{i=1}^{t-1}(i_j-1)+i_a=d-1-t$. This implies that there are at most $q^{d-1-t}$ $\bh_{a,b}$s lying in the space spanned by these
$d-2$ columns.

It remains to count the number of distinct $d-2$ column sets.
Note that
$B_a$ is always selected. Thus, we only have at most $\binom{a-1}{t-1}\leq \binom{\frac{n}{r+1}}{t-1}\leq (\frac{n}{r+1})^{t-1}$ combinations of these $t$ blocks.
After fixing these $t$ blocks, there are at most $(t(r+1))^{d-2}$ ways to pick  $d-2$ columns from these $t$ blocks due to the fact that these $t$ blocks contain only $t(r+1)$ columns. In total, there are at most $(t(r+1))^{d-2}(\frac{n}{r+1})^{t-1}$ ways to
pick $d-2$ columns that precede the $(i,j)$-th column.
Each combination contributes to at most
$q^{d-1-t}$ bad $\bh_{a,b}$. Thus, the number of bad $\bh_{a,b}$ are upper bounded by
\begin{eqnarray*}
\sum_{t=1}^{\lfloor\frac{d-1}{2}\rfloor}\big(t(r+1)\big)^{d-2}\left(\frac{n}{r+1}\right)^{t-1}q^{d-1-t}
&\leq&\Big(\frac{q(d-1)(r+1)}{2}\Big)^{d-2}\sum_{t=1}^{\lfloor\frac{d-1}{2}\rfloor}\Big(\frac{n}{q(r+1)}\Big)^{t-1}\\
&\leq &\Big(\frac{q(d-1)(r+1)}{2}\Big)^{d-2} \Bigl( \frac{d-1}{2}\Bigr) \Bigl(\frac{n}{q(r+1)}\Bigr)^{\lfloor\frac{d-3}{2}\rfloor}.
\end{eqnarray*}
The first inequality is due to $t\leq \frac{d-1}{2}$ and
the last inequality is due to $n>q(r+1)$.
Plug $n=\eta q^{1+\frac{1}{\lfloor (d-3)/2 \rfloor}}$ into the formula. This number is upper bounded by $q^{d-1}(d-1)^{d-1}(r+1)^{(d-2)/2} \eta^{\lfloor (d-3)/2 \rfloor}$; by picking $\eta$ small enough as a function of $d,r$ we can ensure this quantity is at most $q^{d-1}/2$.

  On the other hand, according to Algorithm $1$, the $a$-th component of $\bh_{a,b}$ should be nonzero. Moreover, the first $\frac{n}{r+1}$ components except $a$-th component are all zero.
 That means, the whole space of $\bh_{a,b}$ is of size $q^{d-1}-q^{d-2}> \frac{1}{2}q^{d-1}$.
 Thus, there always exists $\bh_{a,b}$ satisfying our algorithm's requirement.

We are almost done. Let $C$ be the code whose parity-check matrix is $H$. It is clear that $C$ has locality $r$. Since any $d-1$ columns of $H$ are linearly independent, $C$ has minimum distance $d(C)$ at least $d$. Because $H$ has $\frac{n}{r+1}+d-2$ rows, the dimension of $C$ is $k(C)\geq n-\frac{n}{r+1}-(d-2)=\frac{rn}{r+1}-(d-2)$. This implies
$$
\frac{k(C)}{r}\geq \frac{n}{r+1}-\frac{d-2}{r}
$$
We divide it into two case.
\begin{itemize}
\item If $d-2<r$,
 the condition $r+1|n$ implies $\left\lceil\frac{k(C)}{r}\right\rceil\geq \frac{n}{r+1}$ and thus $k(C)+\left\lceil\frac{k(C)}{r}\right\rceil\geq n-d+2$.
It follows that
$$
d(C)\geq d\geq n-k(C)-\left\lceil\frac{k(C)}{r}\right\rceil+2.
$$
Thus, $C$ is an optimal LRC. We are done.
\item If $d-2=r$,
the condition $r+1|n$ implies $\left\lceil\frac{k(C)}{r}\right\rceil\geq \frac{n}{r+1}-1$ and thus $k(C)+\left\lceil\frac{k(C)}{r}\right\rceil\geq n-d+1$.
It follows that
$$
d(C)\geq d\geq n-k(C)-\left\lceil\frac{k(C)}{r}\right\rceil+1.
$$
$C$ is still an optimal LRC because there does not exist LRC reaching the Singleton-type bound. \qedhere
\end{itemize}
\end{proof}

Next, we extend this theorem to the case $n$ is not divisible by $r+1$ and $d\leq r+2$.
\begin{cor}\label{cor:shorten}
Assume $n \equiv a \pmod {r+1}$ and $a>d-1$. There exists optimal LRC of length $n=\Omega_{d,r}(q^{1+\frac{1}{\lfloor (d-3)/2\rfloor}})$. In particular, one obtains the best possible length $n=O(q^2)$ for optimal LRC of minimum distance $5$ if $n  \pmod {r+1}>4$.
\end{cor}
\begin{proof}
Let $N$ be the smallest integer that $N\geq n$ and $r+1|N$, i.e., $N=n+r+1-a$. We construct the parity-check matrix $H$ by
running Algorithm 1. By Theorem~\ref{thm:construction}, $H$ is an $(\frac{N}{r+1}+d-2)\times N$ matrix and any $d-1$ columns of $H$ are linearly independent. We remove the last $N-n=r+1-a$ columns from $H$ and denote the resulting matrix by $H_1$.
Let $C$ be the linear code derived from parity-check matrix $H_1$. Then, $C$ has length $n$, minimal distance $d(C)\geq d$,
locality $r$ and dimension $k(C)\geq n-\frac{N}{r+1}-(d-2) = \frac{r N}{r+1} - r + (a-d+1)$, so that
$$\frac{k(C)}{r}\geq\frac{N}{r+1}-1+\frac{a-(d-1)}{r}.$$
As $a-(d-1)>0$, this gives us $\left\lceil\frac{k(C)}{r}\right\rceil\geq \frac{N}{r+1}$ and thus $k(C)+\left\lceil\frac{k(C)}{r}\right\rceil
\geq n-(d-2)$.
It follows that
$$
d(C)\geq d\geq n-k(C)-\left\lceil\frac{k(C)}{r}\right\rceil+2
$$
Thus, $C$ is an optimal locally repairable code.
\end{proof}

Under some reasonable assumption, the LRCs that are slightly away from the Singleton-type bound might be optimal.
\begin{lemma}{\cite[Theorem III.3]{KBTY18}}
\label{lm:opt}
Assume that  $C$ has $\lceil \frac{n}{r+1} \rceil$ disjoint recovery sets, a linear code $C$ with length $n=a \bmod r+1$, $a\neq 0,1$ and dimension either $k \bmod r \geq a$ or $r|k$ must obey that $d\leq n-k-\lceil\frac{k}{r}\rceil+1$.
\end{lemma}
With the help of this lemma, we can extend the optimality of the shortened LRCs in Corollary~\ref{cor:shorten} to cover almost all the parameters for $d\leq r+2$.
\begin{cor}\label{cor:opt}
Assume $n \equiv a \pmod {r+1}$ and $a\neq 1$. There exists optimal LRC of length $n=\Omega_{d,r}(q^{1+\frac{1}{\lfloor (d-3)/2\rfloor}})$.
\end{cor}
\begin{proof}
It suffices to consider the case that $a\leq d-1$ as Corollary \ref{cor:shorten} already covers the rest of the case.
The construction of the optimal LRCs is the same as that in Corollary \ref{cor:shorten}. It is easy to check that
such code $C$ satisfies
\begin{equation}\label{eq:subopt}
d(C)\geq n-k(C)-\lceil\frac{k(C)}{r}\rceil+1,
\end{equation}
no matter what value $a$ takes.
To see the optimality of this code, it suffices to check whether $k(C)$ meets the condition of Lemma \ref{lm:opt}.
We divide it into two cases. We first consider $a<d-1$.
Observe that the equality in \eqref{eq:subopt} indicates $k(C)= n-\frac{N}{r+1}-(d-2)=\frac{rN}{r+1}-(N-n+d-2)$. This implies
$$k(C) \equiv -(r+1-a+d-2) \equiv a-d+1+r \pmod {r } $$
due to $N=n+r+1-a$ and $a<d-1$. The desired result follows from $k(C) \pmod r \equiv a-d+1+r >a$.
We turn to the second case $a=d-1$. In this case, the similar argument leads to $k(C)\equiv -(r+1-a+d-2)\equiv -(1-a+d-2)\equiv 0 \pmod r$. The desired result follows from the fact that $r| k(C)$.
\end{proof}

\end{document}